\documentclass[11pt]{article}
\usepackage{amsmath}
\usepackage{graphicx}
\usepackage{enumerate}
\usepackage{url} 
\usepackage{enumitem}
\usepackage{amsmath, amsthm, amssymb}
\usepackage{listings}
\usepackage{graphicx}
\usepackage{grffile}
\usepackage[makeroom]{cancel}
\usepackage{adjustbox} 
\usepackage{float}
\usepackage{setspace}
\usepackage{afterpage}
\usepackage{multirow}
\usepackage{xcolor}

\usepackage{bm}         

\usepackage{url} 

\usepackage{bbm} 
\usepackage{placeins} 

\theoremstyle{plain}

\newtheorem{remark}{Remark} 

\usepackage{mathrsfs} 

\usepackage{caption}
\usepackage{subcaption} 

\usepackage{comment}
\specialcomment{itcomment}{\begingroup\itshape}{\endgroup}

\usepackage[ruled,vlined,linesnumbered]{algorithm2e} 

\usepackage{tikz} 
\usetikzlibrary{arrows,shapes,trees} 

\newtheorem{proposition}{Proposition}

\usepackage[parfill]{parskip}
\usepackage{vmargin}
\setmargrb{20mm}{20mm}{20mm}{20mm} 

\usepackage{comment}
\specialcomment{itcomment}{\begingroup\itshape}{\endgroup}

\usepackage{tkz-graph}  
\usetikzlibrary{shapes.geometric}%
\usetikzlibrary{calc}

\usepackage{breakcites}

\begin{document}




\title{\vspace*{-2cm}\sc An $n$-dimensional Rosenbrock Distribution for MCMC Testing}


\author{Filippo Pagani\footnote{Department of Mathematics, The University of Manchester, Manchester, M13 9PL, UK.} \footnote{Corresponding Author: filippo.pagani@manchester.ac.uk} \and Martin Wiegand$^*$ \and Saralees Nadarajah$^*$
}
\maketitle

%
\begin{abstract}
The Rosenbrock function is an ubiquitous benchmark problem for numerical optimisation, and variants have been proposed to test the performance of Markov Chain Monte Carlo algorithms. In this work we discuss the two-dimensional Rosenbrock density, its current $n$-dimensional extensions, and their advantages and limitations. 
We then propose a new extension to arbitrary dimensions called the Hybrid Rosenbrock distribution, which is composed of conditional normal kernels arranged in such a way that preserves the key features of the original kernel. Moreover, due to its structure, the Hybrid Rosenbrock distribution is analytically tractable and possesses several desirable properties, which make it an excellent test model for computational algorithms. 


\end{abstract}
\noindent%


\FloatBarrier

\section{Introduction}

\label{Int}

The Rosenbrock function is a popular test problem in the optimisation literature \cite{rosenbrock1960} due to its challenging features: its minimum is located at the bottom of a long and narrow parabolic valley. 
The original function can be turned into a probability density that maintains these features, and has been adopted by the Markov chain Monte Carlo (MCMC) community to serve as a benchmark problem when testing MCMC algorithms (see for example \cite{goodman2010}). 

The first MCMC method dates back to the 1950s, when the scientists working on the Manhattan Project in Los Alamos used a Random Walk Metropolis (RWM) to simulate systems of particles \cite{metropolis1953}. After that, MCMC remained largely confined to the physics literature until the 1990s, when \cite{gelfand-smith1990} popularised it to the statistics community. This spurred a new wave of research efforts that yielded advanced algorithms such as the Metropolis-adjusted Langevin algorithm (MALA) \cite{roberts1997}, Reversible Jumps MCMC \cite{green1995}, Hamiltonian Monte Carlo (HMC) \cite{duane1987, neal2010} 
among others. See \cite{robert-casella2011} for a historical perspective. 

One of the current frontiers of research in this field is developing algorithms (e.g.\ \cite{girolami2011} and \cite{parno2014}) that can sample efficiently from densities that have 2-$d$ marginals with non-constant or curved correlation structure (see e.g.\ Figure \ref{2drosenIm}). Such shapes make it difficult for MCMC algorithms to take large steps, increasing the autocorrelation time and decreasing the quality of the MCMC sample.

Distributions with curved correlation structures often arise when dealing with complex or hierarchical models, typically found in cosmology \cite{des2017}, epidemiology \cite{house2016}, chemistry \cite{cotter2019}, finance \cite{kim1998}, biology \cite{christensen2003,sullivan2010}, ecology \cite{rockwood2015}, particle physics \cite{feroz2008,allanach2007} and many other subject areas. Sometimes reparametrising the model can map the problematic components to more linear shapes, but it is not always possible, and the reparametrisation may not solve the problem entirely.

Researchers developing new methods for distributions with non-linear correlation structure often test their algorithms on only a handful of benchmark models, amongst which the (2-dimensional) Rosenbrock kernel is quite popular \cite{hogg2018}. However, few properties of the Rosenbrock kernel have been investigated and formalised, especially regarding multivariate extensions of the density for the purpose of MCMC sampling. As we will show in Section \ref{fullrosen}, sometimes the properties of this distribution are so poorly understood that extending the kernel from two to three dimensions radically changes its shape.

In our view, a good benchmark model must possess five main properties. $i$) It must have complex marginal density structure. $ii$) It must be easily extendable to arbitrary dimensions. $iii$) The normalising constant must be known. $iv$) The effect that the parameters have on the shape of the distribution must be clear. $v$) It must be easy to obtain a direct Monte Carlo sample from the test distribution.
A benchmark model that possesses these properties can be tweaked to obtain exactly the shape and features desired, and it provides analytical solutions and large $\emph{iid}$ samples that can be compared with a sample drawn from the same model using the MCMC algorithm being tested. Conversely, a model that behaves uncontrollably when changing its parameters is a terrible benchmark model.

In this work we present the Hybrid Rosenbrock distribution, a benchmark model that possesses all the properties outlined above. The Hybrid Rosenbrock distribution can be used by researchers developing new MCMC methods to test how algorithms perform on distributions with curved 2-$d$ marginals. However, the Hybrid Rosenbrock distribution could also be used by savvy MCMC practitioners to perform algorithm selection. The shape and features of the Hybrid Rosenbrock can be tweaked to match those of the model of interest, which would provide the practitioners with a tailor made toy problem to test their algorithm of choice, and assess how well it performs when compared with the true solution. 

Moreover, the Hybrid Rosenbrock distribution can be used to test the accuracy of algorithms that 
estimate the normalising constant of a kernel \cite{gelfand-smith1990, satagopan2000}. Prominent approaches include \cite{chib1995, diciccio1997} and \cite{jasra2006}, among various other contributions. Due to the number of approaches suggested, having a challenging benchmark problem for which the normalising constant is known would prove a valuable assessment tool.

The structure of the paper is as follows. In Section \ref{sec:tools} we review the main computational tool we used in this work, i.e.\ the simplified Manifold MALA algorithm. In Section \ref{currentlit} we review the current literature on 2-$d$ Rosenbrock distributions and available $n$-dimensional extensions. In Section \ref{normalis} we show how to calculate the normalising constant for the 2-$d$ case, and why it cannot be calculated in the same way for the other variants of the Rosenbrock kernel in the literature. In Section \ref{hybridrosen} we present our $n$-dimensional extension, and discuss how it improves on the shortcomings of current solutions. In Section \ref{numtest} we discuss how changes in the structure and shape of the Hybrid Rosenbrock densitiy affect how challenging it is to obtain an MCMC sample from it.

\section{Tools for testing}
\label{sec:tools}

As none of the distributions listed in Section \ref{currentlit} has known normalisation constant, MCMC methods were used to infer their shapes and produce the figures. All the distributions described in this work have very peculiar features, i.e.\ they look like thin and elongated curved ridges, which are usually problematic to explore for MCMC algorithms. For this reason we selected a state of the art algorithm to perform our analysis, which is one of the few MCMC methods that can successfully cope with those features: the Simplified Manifold MALA (sMMALA) \cite{girolami2011}. 

The sMMALA algorithm, based on the MALA algorithm \cite{roberts1997}, is part of a class of methods that use local information about the target when proposing a move in the state space. 
The sMMALA algorithm will propose a new position $\mathbf{x}'$ in the state space from the current position $\mathbf{x}$ according to the equation
\begin{equation}
\label{smmala}
\mathbf{x}' = \mathbf{x} + \frac{h}{2} \Sigma(\mathbf{x}) \nabla \log \pi(\mathbf{x}) + \mathcal{N}_n (0, h \, \Sigma(\mathbf{x})) \, , \qquad \mathbf{x} \in \mathbb{R}^n \, .
\end{equation}
Here $\pi(\mathbf{x})$ is the distribution of interest, $\nabla$ represents the gradient operator, $\Sigma(\mathbf{x})$ is a positive definite matrix, and $h \in \mathbb{R}^+$ is the step size of the algorithm, parameter tuned by the user to achieve the desired level of acceptance. The proposed $\mathbf{x}'$ then is accepted with a Metropolis acceptance/rejection step, which ensures that the sMMALA sample comes from the correct stationary distribution $\pi(\mathbf{x})$.


A common choice of $\Sigma(\mathbf{x})$ is the Fisher Information matrix (i.e. the negative expectation of the Hessian of the log-likelihood) \cite{girolami2011}, as it carries information on the local correlation structure of the target. In our case the most convenient choice of $\Sigma(\mathbf{x})$ is given in \cite{betancourt2013}, which uses a regularised version of the Hessian of the log-density derived by multiplying its eigenvectors by the absolute value of the eigenvalues. If the eigenvalues are too small, the eigendecomposition may be unstable, so the algorithm regularises the Hessian further by increasing the problematic eigenvalues by a factor of $1/\alpha$, where $\alpha$ is a user defined parameter.

In the rest of this work, sMMALA will be our main tool to infer the shape of a distribution and perform computational tests.

\section{Current literature}
\label{currentlit}

\subsection{The 2-$d$ Rosenbrock distribution}
\label{rosen1}

\begin{figure}[!h] 
	\centering
	\includegraphics[scale=0.6]{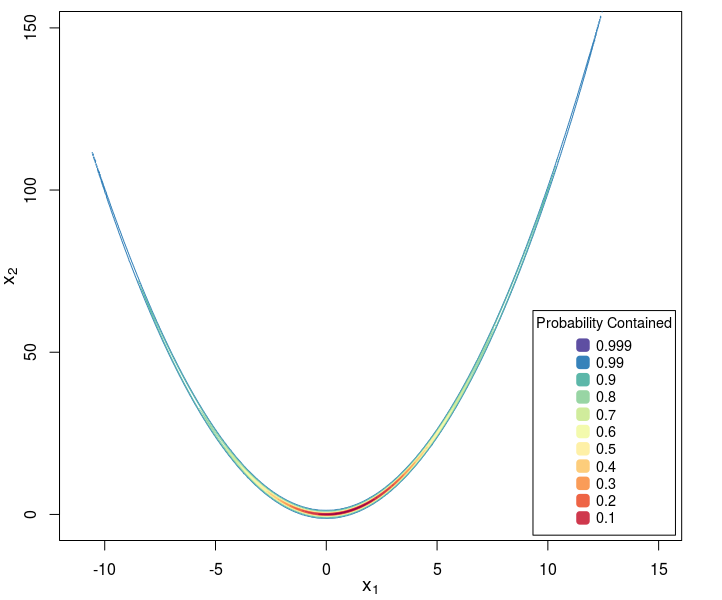}
	\caption{ Contour plot of the 2-$d$ Rosenbrock density as described in Equation \eqref{2drosen}. 
	}
	\label{2drosenIm}
\end{figure}

The simplest non-trivial case of the Rosenbock distribution is the 2-$d$ case, where the kernel can be written as 
\begin{equation}
\label{2drosen}
\pi(x_1,x_2) \propto \exp \left \{ - [ 100 \, (x_2-x_1^2)^2 + (1-x_1)^2]/20 \right \} \, , \qquad x_1,x_2 \in \mathbb{R} \, .
\end{equation}
We follow \cite{goodman2010} when rescaling Equation \eqref{2drosen} by $1/20$, so that the distribution takes the shape of a curved narrow ridge -- shown in Figure \ref{2drosenIm} -- which is normally quite challenging for MCMC algorithms to explore.

It is not clear from the literature how the shape of the kernel in \eqref{2drosen} is affected by changes in the coefficients. Moreover, the normalising constant is generally unknown, and there is more than one way to extend the distribution beyond two dimensions. Two methods have been proposed in the literature, and we will review them to point out their advantages and limitations.

\subsection{Full Rosenbrock distribution}
\label{fullrosen}

We will refer to the $n$-dimensional extension in \cite{goodman2010} as ``Full Rosenbrock" kernel in the following paragraphs. 
The kernel
has the following structure:
\begin{equation}
\label{fullRosen}
\pi(\mathbf{x}) \propto \exp \left \{ - \sum_{i=1}^{n-1} \left [ 100 \, (x_{i+1}-x_i^2)^2 + (1-x_i)^2 \right ]/20 \right \} \, \qquad \mathbf{x}=[x_1,\dots,x_n]^\top \in \mathbb{R}^n \, .
\end{equation}
The normalising constant is unknown. In three dimensions the kernel above can be written as
\begin{equation}
\label{3dfull-rosen}
\pi(\mathbf{x}) \propto \exp \left \{ - \left [ 100 \, (x_2-x_1^2)^2 + (1-x_1)^2 + 100 \, (x_3-x_2^2)^2 + (1-x_2)^2 \right ]/20 \right \} \, \qquad \mathbf{x} \in \mathbb{R}^3 \, .
\end{equation}
Figure \ref{2drosen-fullGW} shows contour plots of a 2 million sample obtained running a sMMALA algorithm on Equation \eqref{3dfull-rosen}, with starting point $\mathbf{x}=[1,\ldots,1]^\top$, $h=1.5$, and $\alpha=10^6$.

\begin{figure}[!h] 
	\includegraphics[scale=0.7]{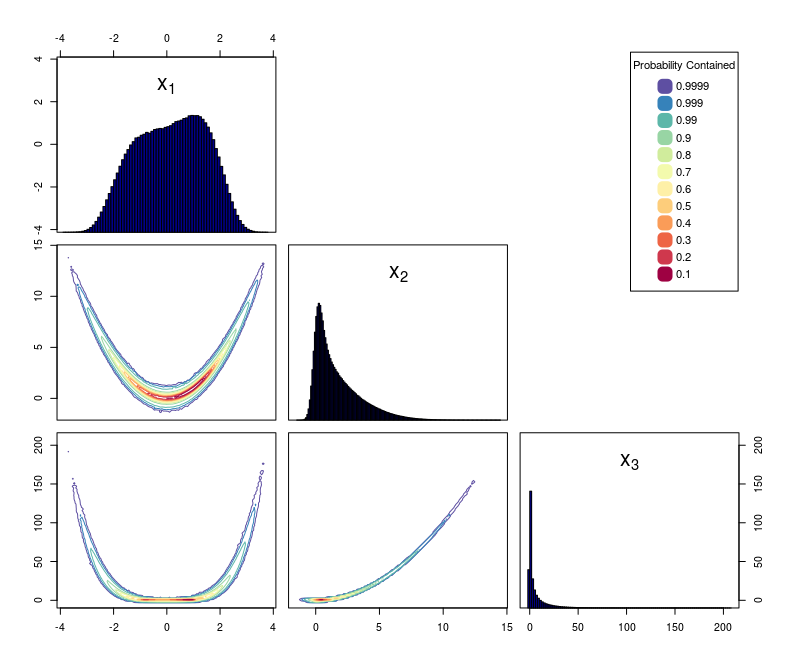}
	\centering
	\caption{Contour plot of a 3-$d$ Full Rosenbrock density, as described in Equation \eqref{3dfull-rosen}, obtained from a sMMALA MCMC sample. }
	\label{2drosen-fullGW}
\end{figure}

We want to draw attention to the joint distribution of the first two random variables $x_1,x_2$ from Figure \ref{2drosen-fullGW}, which we show in more detail in Figure \ref{2drosenCI}, and we want to compare them to the same two variables from the 2-$d$ Rosenbrock kernel. 
Evidently, extending the kernel from a 2-$d$ Rosenbrock to a 3-$d$ Full Rosenbrock significantly changes the joint plot between the variables $x_1$ and $x_2$: the long narrow ridge
has become much more concentrated around the mode. 
Moreover, the specific change in shape significantly reduces the difficulty of sampling from the distribution via MCMC methods, and is directly against one of our requirements for a good test problem.

\begin{figure}[!h] 
	\includegraphics[scale=0.6]{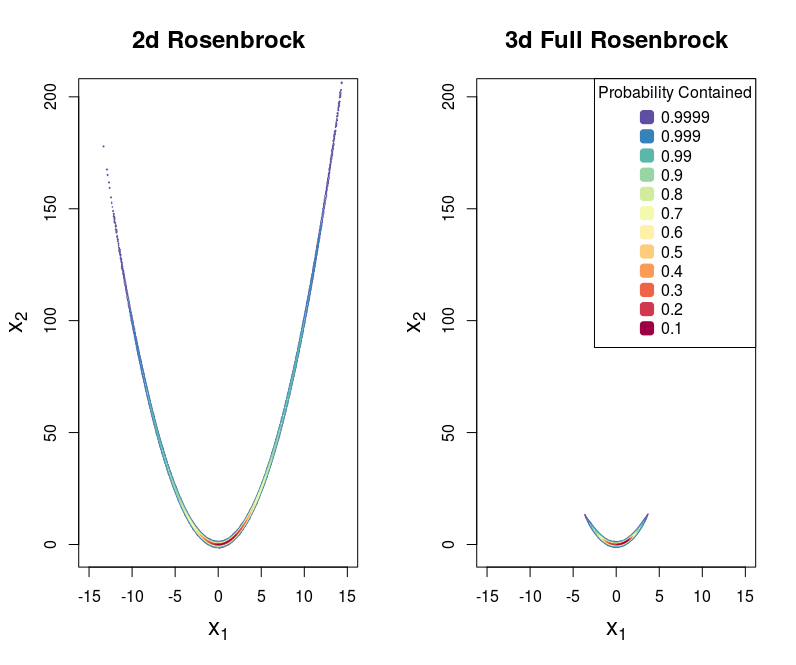}
	\centering
	\caption{Contour plots of a 2-$d$ Rosenbrock density as described in Equation \eqref{2drosen}, and of the $x_1$ and $x_2$ variables from a 3-$d$ Full Rosenbrock kernel from Equation \eqref{3dfull-rosen}.}
	\label{2drosenCI}
\end{figure}

However, the Full Rosenbrock kernel does have some desirable features: as $n$ increases, the variance of $x_n$ increases steeply, as each new random variable is directly dependent on the squared value of the previous variable. Densities with such properties (e.g. Neal's Normal \cite{neal2010}) are known to pose a challenge to MCMC algorithms. 

It should be noted that with variances increasing steeply as $n$ increases (depending on the choice of parameters), the Full Rosenbrock kernel may not be extended beyond a certain number of dimensions, as the numerical computations may become unstable. Model 5 in Section \ref{sensitivityAnalysis} suffers from the same drawback. One possible solution to this problem is adopting a ``block structure", a feature typical of the Even Rosenbrock kernel described in the next section.




\subsection{Even Rosenbrock distribution}
\label{evenrosen}

In the optimisation literature, \cite{dixon1994} proposes a simpler version of the Full Rosenbrock function used in Section \ref{fullrosen}, which can be turned into a kernel as
\begin{equation}
\label{ndrosen-even}
\pi(\mathbf{x}) \propto \exp \left \{ - \sum_{i=1}^{n/2} \left [ (x_{2i-1} - \mu_{2i-1})^2 - 100 \, (x_{2i}-x_{2i-1}^2)^2 \right ] /20 \right \} \, , \qquad \mathbf{x} \in \mathbb{R}^{n}
\end{equation}
where $n$ must be an even number, and we maintain the $1/20$ mentioned in the previous section. The normalising constant is unknown. This density could be described as the product of $n/2$ independent 2-$d$ Rosenbrock kernels. Figure \ref{4drosen-even} shows the shape of the 2-$d$ marginal distributions\footnote{The contours were plotted using a sample from a sMMALA algorithm tuned exactly as described in the previous section, with $\alpha=10^6$, $\mathbf{x}=\underbar{1}$ and acceptance ratio roughly 50\%.} of \eqref{ndrosen-even} when taking $n=4$ and $\mu_1=\mu_3=1$, which result in the kernel
\begin{equation}
\label{5drosen-even}
\pi(\mathbf{x}) \propto \exp \left \{ - \left [ (x_1-1)^2 + 100 \, (x_2-x_1^2)^2 + (x_3-1)^2 + 100 \, (x_4-x_3^2)^2 \right ]/20 \right \} \, .
\end{equation}

\begin{figure}[!h] 
	\includegraphics[scale=0.7]{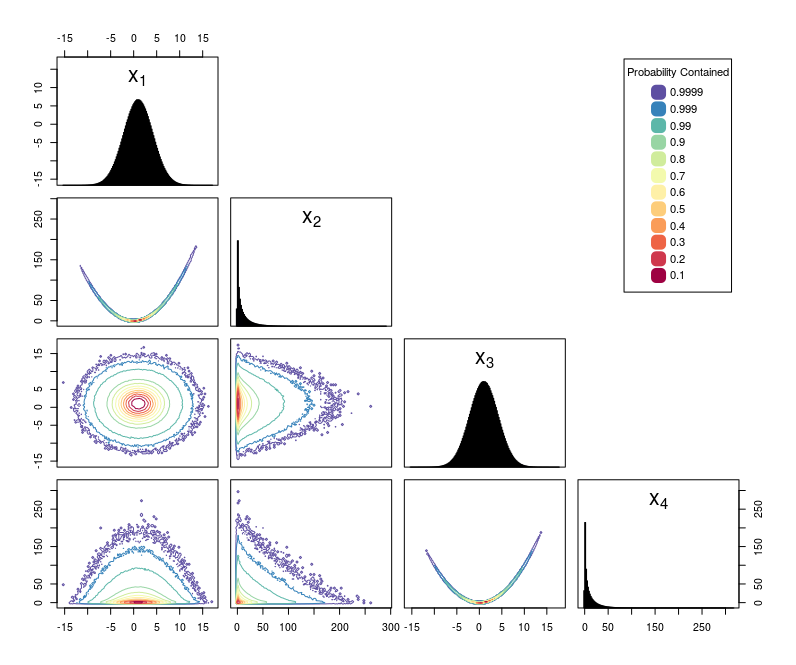}
	\centering
	\caption{Contour plot of a 4-$d$ Even Rosenbrock density, as described in Equation \eqref{ndrosen-even}. Most of the joint distributions are uncorrelated. 
	}
	\label{4drosen-even}
\end{figure}

Equation \eqref{ndrosen-even} represents a more straightforward problem than Equation \eqref{3dfull-rosen}. The round shapes and lack of ridges in the lower left plots of Figure \ref{4drosen-even} (specifically for the pairs $(x_1,x_3), (x_1,x_4), (x_2,x_3)$ and $(x_2,x_4)$ ) confirm the lack of complex dependencies that characterise the Full Rosenbrock kernel. Another important difference is that, unlike the Full Rosenbrock kernel, the Even Rosenbrock does maintain the shape of the joint 2-$d$ marginals as new dimensions are added. Moreover, the variance of the variable $x_n$ remains stable as $n$ grows large.

However, the downside of this model is that only a small fraction of the joint distributions (e.g. the 2-$d$ marginals of $(x_1,x_2)$ and $(x_3,x_4)$ ) will be curved narrow ridges, while the majority of the 2-$d$ marginals will be uncorrelated. Furthermore most of the variables will have similar variances, which also significantly reduces the difficulty of the problem.

\section{Normalising constant and interpretation of the parameters}
\label{normalis}

The normalising constant is unknown for all the examples covered in Section \ref{currentlit}. However, we found that we can rewrite the 2-$d$ Rosenbrock kernel from Equation \eqref{2drosen} in general form as:
\begin{align}
\label{2drosen-norm}
\pi(x_1,x_2) & \propto \exp \left \{ - a (x_1-\mu)^2 - b (x_2-x_1^2)^2 \right \} \nonumber \\
& \propto \exp \left \{ -\frac{1}{2\frac{1}{2a}} (x_1-\mu)^2 - \frac{1}{2\frac{1}{2b}} (x_2-x_1^2)^2 \right \} \, .
\end{align}
where $a=1/20$, $b=100/20$, and $\mu=1$, and more generally $a,b \in \mathbb{R}^+, \mu \in \mathbb{R}, \, x_1,x_2 \in \mathbb{R}$. Equation \eqref{2drosen-norm} should make it obvious that the density is composed of two normal kernels, i.e. $\pi(x_1,x_2) = \pi(x_1) \pi(x_2|x_1)$, where
\begin{equation*}
\pi(x_1) \sim \mathcal{N} \left (\mu, \frac{1}{2a} \right ) \, , \qquad \pi(x_2|x_1) \sim \mathcal{N} \left (x_1^2,\frac{1}{2b} \right ) \, .
\end{equation*}
Interpreting the 2-$d$ Rosenbrock density as the composition of two normal kernels
allows us to calculate the normalisation constant as follows. \\

\begin{proposition}
\label{2drosen-proposition}
The normalisation constant of the 2-$d$ Rosenbrock kernel as shown in Equation \eqref{2drosen-norm} is $\sqrt{ab} / \pi$.
\end{proposition}

\begin{proof}
We begin by integrating Equation \ref{2drosen-norm} over the domain $\mathbb{R}^2$:
\begin{align*}
\int_{-\infty}^\infty \int_{-\infty}^\infty & \exp \left \{ -\frac{1}{2\frac{1}{2a}} (x_1-\mu)^2 - \frac{1}{2\frac{1}{2b}} (x_2-x_1^2)^2 \right \} \, dx_2 \, dx_1 = \\
& = \int_{-\infty}^\infty  \exp \left \{ -\frac{1}{2\frac{1}{2a}} (x_1-\mu)^2 \right \} \int_{-\infty}^\infty \exp \left \{ - \frac{1}{2\frac{1}{2b}} (x_2-x_1^2)^2 \right \} \, dx_2 \, dx_1 \, .
\end{align*}
We can apply a change of variables in the second integral, $v=x_2-x_1^2$, which becomes
\begin{align*}
& = \int_{-\infty}^\infty  \exp \left \{ -\frac{1}{2\frac{1}{2a}} (x_1-\mu)^2 \right \} \int_{-\infty}^\infty \exp \left \{ - \frac{1}{2\frac{1}{2b}} v^2 \right \} \, dv \, dx_1 \, ,
\end{align*}
expression that highlights the two kernels $x_1 \sim \mathcal{N}(\mu, 1/2a)$ and  $v \sim \mathcal{N}(0,1/2b)$. Solving the integrals individually,
\begin{align*}
& = \sqrt {2\pi \frac{1}{2a} } \, \sqrt {2\pi \frac{1}{2b} } \qquad \qquad \qquad \qquad \qquad \qquad \qquad \qquad \qquad \\
& = \frac{\pi}{\sqrt{ab}} \, .
\end{align*}
The reciprocal of this number is the normalisation constant.
\end{proof}

Interpreting the 2-$d$ Rosenbrock density as the composition of two normal kernels
also provides us with a simple interpretation for the coefficients: $1/2a$ is the variance of the first normal distribution in the $x_1$ dimension, while $1/2b$ is the variance of the second normal distribution, which is concentrated in the $x_2$ dimension around the manifold given by the parabola $x_1^2$. Therefore, increasing $2a$ increases the slope of the distribution \emph{along} the parabola, while increasing $2b$ decreases the dispersion \emph{around} the parabola. Naturally, the variances of the marginals will be slightly different from those of the conditionals, as the 2-$d$ distribution has to be projected on the corresponding axes. The parameter $\mu$ determines the position on the mode of the variable $x_1$ along the parabola. In Figure \ref{2drosen-parComparison} we show how changing the parameters $\mu$, $a$ and $b$ influences the shape of the 2-$d$ Rosenbrock distribution.

\begin{figure}[!h] 
	\includegraphics[scale=0.75]{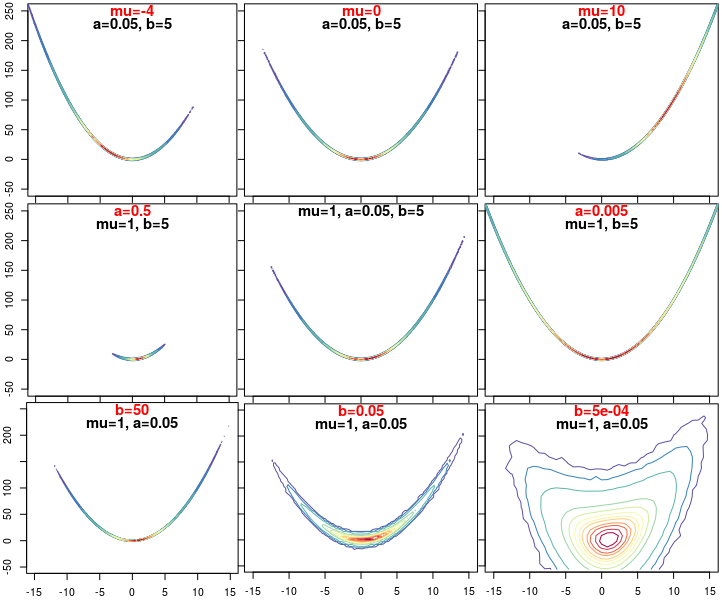}
	\centering
	\caption{Contour plot for the variables ($x_1$, $x_2$) of a 2-$d$ Rosenbrock density, as the parameters $\mu$, $a$ and $b$ take different values. For comparison, the central plot represents a kernel with the original values of the parameters, i.e. $\mu=1$, $a=1/20$ and $b=100/20$.}
	\label{2drosen-parComparison}
\end{figure}

Moreover, the structure of the 2-$d$ problem is such that we can use the chain rule of probability to split the joint density into its conditional densities as factors. We can then sample from each conditional distribution independently, and calculate estimates, credibility regions,  QQ-plots and more. Given these desirable properties, it is highly convenient to find an $n$-dimensional extension of the 2-$d$ Rosenbrock density that preserves this structure. \\ 

\begin{remark}
	In this work we will only investigate the case where the kernels are normal and connected via the mean of the second kernel, and where the mean function is the parabola $x_1^2$. Indeed, any polynomial in $x_1$ can be used as mean function in the $x_2|x_1$ kernel, as well as other functions such as $\exp ( x_1 )$, $\sin ( x_1 )$, $1/x_1$. In fact, any function $f(x_1): \mathbb{R} \to E \subseteq \mathbb{R}$ that does not alter the behaviour of the integrals in the proof of Proposition \ref{2drosen-proposition} is a viable candidate as mean function. Furthermore, as long as the same conditions are satisfied, kernels other than normal can be used. For example, the following joint distribution uses a uniform kernel:
	\begin{equation}
	f(x_1)f(x_2|x_1) = \frac{1}{\sqrt{2 \pi}} e^\frac{-x_1^2}{2}   \frac{1}{x_1^2} \mathbbm{1}_{[0,x_1^2]}(x_2) \, .
	\end{equation}
	These considerations also apply to the $n$-dimensional distribution we propose in Section \ref{hybridrosen}.
\end{remark}

In light of this, we are able to explain why the Full Rosenbrock kernel changes shape as its dimension increases. For simplicity, we will illustrate our point using a 3-$d$ Full Rosenbrock model. From Equation \eqref{3dfull-rosen} we can derive the following general expression:
\begin{equation}
\label{3dfull-rosen2}
\pi(\mathbf{x}) \propto \exp \left \{ - a(x_1-\mu_1)^2 - b (x_2-x_1^2)^2 - c (x_2-\mu_2)^2 - d(x_3-x_2^2)^2  \right \} \, ,
\end{equation}
where $\mathbf{x} \in \mathbb{R}^3$ and $a,b,c,d \in \mathbb{R}^+$, $\mu_1, \mu_2 \in \mathbb{R}$. While the first and fourth terms are two normal kernels in $x_1$ and $x_3$ and can be easily isolated, the $x_2$ kernel is now composed of two terms. Consequently, the integral of \eqref{3dfull-rosen2} with respect to $x_3$, does not yield Equation \eqref{2drosen}. 
In order to obtain a more compact expression for the kernel in the variable $x_2$, we expand the second and third terms of Equation \eqref{3dfull-rosen2} to a sum of monomials, and complete the squares by adding and subtracting the necessary terms: 
\begin{align}
\label{squares} 
- b & (x_2-x_1^2)^2  -c (x_2-\mu_2)^2 = \nonumber \\
%
%
%
& = - \frac{1}{2} \left ( \frac{x_2-\frac{2bx_1^2+2c\mu_2}{2b+2c} }{\frac{1}{\sqrt{2b+2c} } } \right )^2 - \frac{(2bx_1^2+2c\mu_2)^2}{2(2b+2c)} -\frac{1}{2} (2bx_1^4+2c\mu_2^2) \, ,
\end{align}
which can be substituted back in \eqref{3dfull-rosen2}. The first term in \eqref{squares} represents the new normal kernel for $x_2$, i.e.
\begin{equation}
\label{ydist}
x_2|x_1 \sim \mathcal{N} \left (\frac{(2bx_1^2+2c\mu_2)^2}{2b+2c}, \frac{1}{2b+2c} \right ) \, .
\end{equation}
This kernel is not influenced by the $x_2$ variable present in the last term of \eqref{3dfull-rosen2} as it disappears after integrating in the variable $x_3$, as we showed in the proof of Proposition \ref{2drosen-proposition}. The other terms in \eqref{squares} are remaining terms from the calculations that we cannot simply include in the normalising constant, as they depend on $x_1$. This changes the kernel of the variable $x_1$, whose distribution is now unknown and cannot be sampled from directly.
The variable $x_2|x_1$ also changes shape drastically. Looking at Equation \eqref{ydist}, the value of the variance of $x_2|x_1$ changes from $1/2b$ to $1/(2b+2c)$, producing the effect
observed in Figure \ref{2drosenCI}.

These considerations extend to higher dimensions ($n>3$), where every time the dimension of the model is increased to $n+1$ according to the scheme in \eqref{fullRosen}, the kernels of the variables $x_1, \ldots, x_n$ change as described above. \\

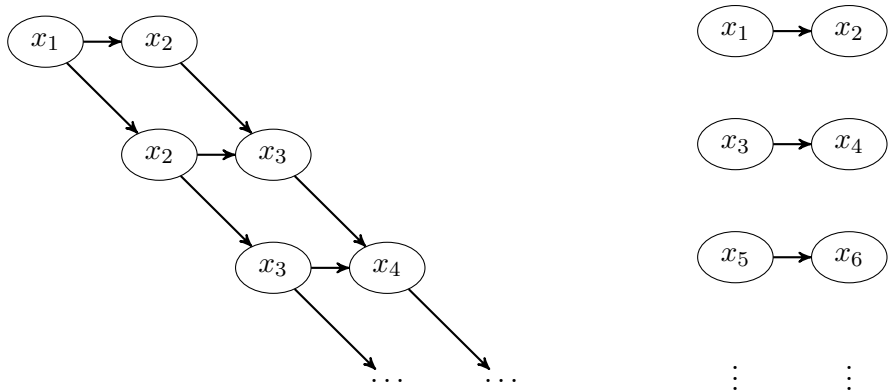
\begin{figure}[ht]
	\center
	\tikzstyle{VertexStyle} = [shape = ellipse, minimum width = 6ex, draw] 
	\tikzstyle{EdgeStyle}   = [->,>=stealth'] 
	\begin{minipage}{.4\textwidth}
	\center
	\begin{tikzpicture}[scale=1.5] 
	\SetGraphUnit{1} 
	\SetVertexMath
	\Vertex[L=x_1]{1} \EA[L=x_2](1){2} 
	\SOEA[L=x_2](1){3} \EA[L=x_3](3){4}
	\SOEA[L=x_3](3){5} \EA[L=x_4](5){6}
	\begin{scope}[VertexStyle/.style = {draw=none}]
	\SOEA[L=\ldots](5){7} \EA[L=\ldots](7){8}
	\end{scope}
	%
	%
	\Edges(1,2) \Edges(1,3,4) \Edges(2,4)
	\Edges(3,5) \Edges(5,6) \Edges(4,6)
	\Edges(5,7) \Edges(6,8) 
	\end{tikzpicture}
	\end{minipage}
	\begin{minipage}{.4\textwidth}
	\center
	\begin{tikzpicture}[scale=1.5] 
	\begin{scope}[node distance=2mm and 10mm]
	\SetGraphUnit{1} 
	\SetVertexMath
	\Vertex[L=x_1]{1} \EA[L=x_2](1){2} 
	\SO[L=x_3](1){3} \EA[L=x_4](3){4}
	\SO[L=x_5](3){5} \EA[L=x_6](5){6}
	\begin{scope}[VertexStyle/.style = {draw=none}]
	\SO[L=\vdots](5){7} \EA[L=\vdots](7){8}
	\end{scope}
	\Edges(1,2) \Edges(3,4) \Edges(5,6) 
	\end{scope}
	\end{tikzpicture}
	\end{minipage}
	\caption{Graphical models representing the dependency structure of the Full Rosenbrock (left) and Even Rosenbrock (right) dependency structure. The circles represent the individual kernels of each variable, while the edges represent the direct dependence between the kernels.}
	
	\label{GMfullNeven}
\end{figure}

Once we understand how the individual kernels depend on each other, we can use a simple graphical model to represent the dependency structure of the Full and Even Rosenbrock kernels, as shown in Figure \ref{GMfullNeven}. The circles represent the kernels of each variable, while the edges represent the direct dependence relationship between the kernels.

\section{Hybrid Rosenbrock distribution} 
\label{hybridrosen}



The overall goal of this paper is to find a $n$-dimensional benchmark model that fulfils the criteria outlined in Section \ref{Int}, i.e. $i$) it must have complex marginal density structure; $ii$) it must be easily extendable to arbitrary dimensions; $iii$) the normalising constant must be known; $iv$) the effect of the parameters on the shape of the distribution must be clear; $v$) it must be easy to obtain a direct Monte Carlo sample from the test distribution.
These properties are vital for a suitable benchmark distribution. 
Furthermore, we want to allow different variables to have very different variances, a property that occurs in the Full Rosenbrock kernel (Section \ref{fullrosen}). This property, coupled with the non-linear correlation structure in the 2-$d$ marginals presents a challenging problem for most MCMC algorithms.

The Hybrid Rosenbock density fulfils all of the outlined criteria, providing a model where every single 2-$d$ marginal distribution has a complex dependency structure. Its kernel can be written as:
\begin{equation}
\label{ndrosen-hybrid2}
\pi(\mathbf{x}) \propto \exp \left \{ - a (x_{1}-\mu)^2 - \sum_{j=1}^{n_2} \sum_{i=2}^{n_1} b_{j,i} (x_{j,i} - x_{j,i-1}^2)^2 \right \} \, ,
\end{equation}
where $\mu, x_{j,i} \in \mathbb{R}$; $a,b_{j,i} \in \mathbb{R}^+$ ($\forall j,i$), and where the final dimension of the distribution is given by the formula $n=(n_1-1)n_2+1$. 

The dependency structure between the components $x_1,\ldots,x_{n_2,n_1}$ of the Hybrid Rosenbrock distribution can be represented with a graphical model, as shown in Figure \ref{GMhybrid}, which can then be compared with the structure of the models from Section \ref{currentlit}, shown in Figure \ref{GMfullNeven}. 

\begin{figure}[ht]
	\center
	\tikzstyle{VertexStyle} = [shape = ellipse, minimum width = 6ex, draw] 
	\tikzstyle{EdgeStyle}   = [->,>=stealth'] 
	
	%
	
	\begin{tikzpicture}[scale=1.5] 
	\SetGraphUnit{1} 
	\SetVertexMath
	\Vertex[L=\; x_1 \;]{1} \NOEA[L=x_{1,2}](1){2} \EA[L=x_{1,3}](2){3} 
	\begin{scope}[VertexStyle/.style = {draw=none}]
	\EA[L=\ldots](3){31}
	\end{scope}
	\EA[L=x_{2,2}](1){4} \EA[L=x_{2,3}](4){5}
	\begin{scope}[VertexStyle/.style = {draw=none}]
	\EA[L=\ldots](5){51}
	\end{scope}
	\begin{scope}[VertexStyle/.style = {draw=none}]
	\SOEA[L=\ldots](1){61}
	\end{scope}
	\Edges(1,2,3,31) \Edges(1,4,5,51) \Edges(1,61)
	\end{tikzpicture}
	\caption{Graphical model representing the dependency structure of the Hybrid Rosenbrock distribution. The circles represent the kernels of each variable, while the edges represent the direct dependence between the kernels.}
	
	\label{GMhybrid}
\end{figure}
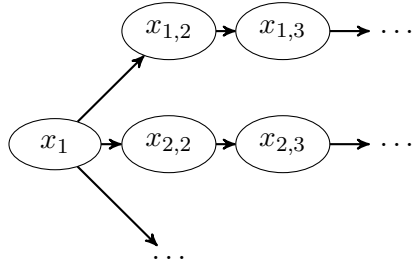
Each ``row" of the diagram in Figure \ref{GMhybrid} represents a ``block of variables". 
The index $i$  in Equation \eqref{ndrosen-hybrid2} identifies variables within a ``block", with $n_1$ denoting the block size, while the index $j$ identifies a single block amongst the $n_2$ blocks present. The indices on the coefficients $b_{j,i}$ follow the block structure, as do the indices on the random variables $x_{j,i}$. The variable $x_{j,1} = x_{1}, \forall j = 1,\dots,n_2$ only has one index, as it is common to all blocks. 

Figure \ref{5drosen-hyb2} shows contour plots obtained from a Monte Carlo sample of the kernel in Equation \eqref{ndrosen-hybrid2}, taking $n_2=2, n_1=3$, and $\mu=1$, $a=1/20$ and $b_{j,i}=100/20$ ($\forall j,i$), i.e.
\begin{align}
\label{5dRosenHyb}
\pi(\mathbf{x}) \propto \exp \bigg \{ - a (x_{1}-\mu)^2 
& - b_{1,2} (x_{1,2} - x_{1}^2)^2 - b_{1,3} (x_{1,3} - x_{1,2}^2)^2 \nonumber \\
& - b_{2,2} (x_{2,2} - x_{1}^2)^2 - b_{2,3} (x_{2,3} - x_{2,2}^2)^2 \bigg \} \, , \qquad \mathbf{x} \in \mathbb{R}^5 \, .
\end{align}

\begin{figure}[!h] 
	\includegraphics[scale=0.75]{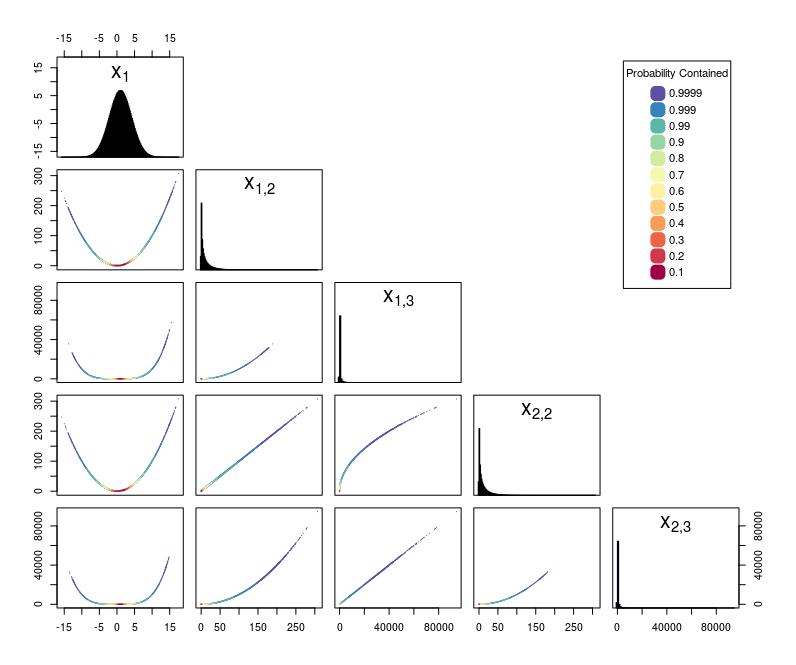}
	\centering
	\caption{Contour plot of a $(n_1, n_2)=(3,2)$ Hybrid Rosenbrock density as described in Equation \eqref{ndrosen-hybrid2}, with parameters $a=1/20$, $b_{j,i}=100/20$ ($\forall j,i$), $\mu=1$, obtained via direct sampling. Every joint distribution is either a straight or curved ridge. }
	\label{5drosen-hyb2}
\end{figure}

The Hybrid kernel inherits from the Full Rosenbrock kernel the feature of having variables with very different variances, as can be observed in the scales of the plots in Figure \ref{5drosen-hyb2}. 
Moreover, as opposed to the Even Rosenbrock contours shown in Figure \ref{4drosen-even}, no plot in Figure \ref{5drosen-hyb2} presents trivial correlation structure: every 2-$d$ marginal is a straight or curved ridge with very long tails.
At the same time, the Hybrid kernel inherits from the Even Rosenbrock kernel the block structure, which makes sure that as $n$ grows, the variance of $x_n|x_{n-1}$ is computationally stable to calculate. 
Notably, all the factors of the distribution are still conditionally or unconditionally normal, which allows us to calculate the normalising constant in the same way as in Proposition \ref{2drosen-proposition}. \\

\begin{proposition}
	\label{constHybridProposition}
	The normalisation constant of the Hybrid Rosenbrock kernel given in Equation \eqref{ndrosen-hybrid2} with $n=(n_1-1)n_2+1$ is 
	\begin{equation*}
	\frac{\sqrt{a} \; \prod_{i=2,j=1}^{n_1,n_2} \sqrt{b_{j,i}} }{\pi^{n/2}} \, .
	\end{equation*}
\end{proposition}

\begin{proof}
The proof is similar to that of Proposition \ref{2drosen-proposition}, where we use the conditional structure of the density to split the integrals of the normal kernels, and solve them one at a time. The details are shown in Appendix \ref{proof:constHybridProposition}.
\end{proof}

Using the conditional normal structure of the model, it is possible to obtain an $iid$ Monte Carlo sample from the joint distribution by using the following algorithm.

\begin{center}
	\begin{algorithm}[H]
		\setstretch{1.5}
		\SetAlgoLined
		\For{$k=1, \dots, N$}{
			$X_{1} \sim \mathcal{N} \left( \mu, \frac{1}{2 a} \right )$\\
			\For{$j=1, \ldots, n_2$}{
				\For{$i=2, \ldots, n_1$}{
					$X_{j,i} | X_{j,i-1} \sim \mathcal{N} \left( x_{j,i-1}^2, \frac{1}{2 b_{ji} } \right)$ \\	
				}	
			}
			$X_{(k)} = (X_{(1)}, X_{1,2}, \dots, X_{n_2,n_1})$
		}

		\Return $\left(X_{(1)},\dots,X_{(N)} \right)$ 
		\label{alg:1}
		\caption{Pseudocode to sample from a Hybrid Rosenbrock Distribution}
	\end{algorithm}
	
\end{center}

\section{Numerical tests}
\label{numtest}

In this section we conduct numerical experiments to complement our theoretical statements made so far. In Section \ref{modelValidation} we assess whether a well tuned MCMC algorithm converges to the theoretical results we obtained in Section \ref{hybridrosen}. In Section \ref{sensitivityAnalysis} we perform empirical tests on how the values of the parameters of the Hybrid Rosenbrock distribution influence the performance of MCMC algorithms sampling from it.

\subsection{Model validation}
\label{modelValidation}

We performed our validation tests on the $(n_1, n_2)=(3,2)$ Hybrid Rosenbrock distribution described in \eqref{5dRosenHyb}, i.e.
\begin{align}
\label{5dRosenHyb2}
\pi(\mathbf{x}) \propto \exp \bigg \{ - a (x_{1}-\mu)^2 
& - b_{1,2} (x_{1,2} - x_{1}^2)^2 - b_{1,3} (x_{1,3} - x_{1,2}^2)^2 \nonumber \\
& - b_{2,2} (x_{2,2} - x_{1}^2)^2 - b_{2,3} (x_{2,3} - x_{2,2}^2)^2 \bigg \} \, , \qquad \mathbf{x} \in \mathbb{R}^5 \, ,
\end{align}
with $\mu=1$, $a=1/20$ and $b_{j,i}=100/20$, $i=2,3$, $j=1,2,3$. We selected this specific target as without being overly challenging for our computational resources, it presents all the main features of the Hybrid Rosenbrock distribution: it is composed of multiple blocks with multiple variables per block. We compared a sample drawn from the kernel above following Algorithm 1, with a sample from the same distribution drawn using a sMMALA algorithm.

Algorithm 1 was run to obtain 2 million samples, while the algorithm sMMALA was run with $\alpha=10^6$, step size $h=.3$, for $20$ million samples. The step size $h$ was chosen so that sMMALA would achieve an acceptance rate of roughly $50\%$. For computational reasons, we then reduced the final number of sMMALA samples to 2 millions by discarding nine out of each ten samples that we obtained from the algorithm sMMALA. The QQ-plots for each variable in Equation \eqref{5dRosenHyb2} are showed in Figure \ref{qq5dNasty}.

\begin{figure}[!h] 
	\includegraphics[scale=0.7]{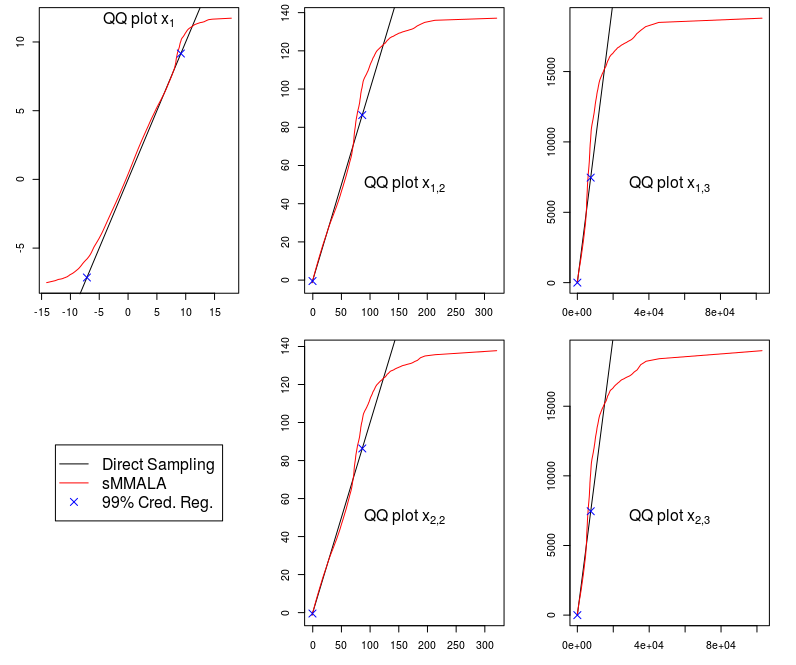}
	\centering
	\caption{QQ-plots for each variable of Equation \ref{5dRosenHyb2}. The horizontal axis show the quantiles obtained from direct Monte Carlo sampling, while the vertical axis shows the quantiles calculated from the sMMALA MCMC sample.}
	\label{qq5dNasty}
\end{figure}

The top left plot in Figure \ref{qq5dNasty} shows the QQ-plot for the variable $x_1$. The red line, representing the empirical quantiles obtained from the sMMALA sample, closely follows the black line, which represents the empirical quantiles calculated from Algorithm 1. Even more so when taking into consideration the 99\% credibility region, where the black and red lines almost completely overlap. 

The variables $x_{1,2}$ and $x_{2,2}$, shown in the middle plots, have tails that stretch moderately far from the mode. Again, the algorithm sMMALA agrees quite well with the sample from Algorithm 1: the red line diverges from the black only near the upper edge of the plot, much further away from the mode than the blue cross representing the empirical 99\% credibility region of the density.

The last two plots on the right side of Figure \ref{qq5dNasty} show the QQ-plots for variables $x_{2,2}$ and $x_{2,3}$, which have tails that reach very far from the mode. Once again, the results from sMMALA and Algorithm 1 are in agreement.

The only discrepancy between the quantiles of Algorithm 1 and sMMALA is in the farthest areas of the tails region of the target. This is due to MCMC algorithms having troubles visiting the tails and returning to the mode efficiently, while direct Monte Carlo sampling does not suffer from this drawback.

In order to control for the Monte Carlo error originating from Algorithm 1, we repeated the same experiment with four million samples taken from Algorithm 1, instead of two. The results in Figure \ref{qq5dNasty} did not change significantly, leading us to believe that the Monte Carlo error that Algorithm 1 introduces in our analysis is negligible.

\subsection{Sensitivity analysis}
\label{sensitivityAnalysis}

In this section we investigate how varying the parameters $n_1, n_2, \mu, a, b_{j,i}$ $(\forall j,i)$ influences the performance of MCMC algorithms sampling from the Hybrid Rosenbrock distribution. 

Our analysis consists in comparing the integrated autocorrelation time $\tau$ calculated using MCMC samples from models with different sets of parameters. The value of $\tau$ roughly measures how many steps on average an MCMC algorithm has to take from an initial position $\mathbf{x}$ before it returns a sample that is completely uncorrelated with $\mathbf{x}$. As we test the same MCMC algorithm on different models, studying how $\tau$ varies for each model provides insights into how easy it is to sample from that model via MCMC (see \cite{goodman2010} and references therein).

To obtain the MCMC samples we rely on a sMMALA algorithm, tuned with $\alpha=10^6$ and acceptance ratio fixed at roughly 50\%. As all our models are multidimensional, each MCMC sample yields a vector of $n$ autocorrelation times $\tau_i$, one for each component of the state space, where $\tau_i$ is defined as
\begin{equation*}
\tau_i = 1 + 2 \sum_{l=1}^L \mbox{corr} (y_i^0, y_i^l ), \quad i=1,\ldots,n \, ,
\end{equation*}
where $\mathbf{y}_i$ is the MCMC sample from the $i$th component of the state space, and $L$ is an integer number representing the last lag where the sample autocorrelation is significantly different from zero. 
We then record only the highest autocorrelation time amongst all components:
\begin{equation*}
\tau = \max_{i=1,\ldots,n} \tau_i \, .
\end{equation*}
Naturally, the smaller the value of $\tau$, the better the algorithm mixes. Assuming the autocorrelation in an MCMC sample is always non-negative, an algorithm that generates $\emph{iid}$ samples achieves the smallest possible value of $\tau$, i.e. $\tau=1$.


In the remainder of this section, we will test six separate distribution structures or models, indexed by the parameters $n_1, n_2$, and for each of them, we will vary the parameters $\mu, a, b_{j,i}$ $(\forall j,i)$ to change the model's shape. For simplicity, we will fix $b_{j,i} = b$, $\forall j,i$. The six models are represented in Figure \ref{GM2drosen}.


\begin{figure}[ht]
	\center
	\tikzstyle{VertexStyle} = [shape = ellipse, minimum width = 6ex, draw] 
	\tikzstyle{EdgeStyle}   = [->,>=stealth'] 
	
	\begin{minipage}[t]{.5\textwidth}
	\center
	\begin{tikzpicture}[scale=1.5] 
	\SetGraphUnit{1} 
	\SetVertexMath
	\Vertex[L=\; x_{1} \;]{1} \EA[L=x_{1,2}](1){2}
	\Edges(1,2)
	\end{tikzpicture}
	\caption*{Model 1: $n_2=1$, $n_1=2$.}
	\end{minipage}
	\vspace{0.1cm}
%
	\begin{minipage}[t]{.4\textwidth}
	\center
	\begin{tikzpicture}[scale=1.5] 
	\SetGraphUnit{1} 
	\SetVertexMath
	\Vertex[L=\; x_{1} \;]{1} \EA[L=x_{1,2}](1){2} \SOEA[L=x_{1,3}](1){3} 
	\Edges(1,2) \Edges(1,3)
	\end{tikzpicture}
	\caption*{Model 2: $n_2=2$, $n_1=2$.}
	\end{minipage}
	\vspace{0.1cm}
%
	\begin{minipage}[t]{.5\textwidth}
	\center
	\begin{tikzpicture}[scale=1.5] 
	\SetGraphUnit{1} 
	\SetVertexMath
	\Vertex[L=\; x_{1} \;]{1} \EA[L=x_{1,2}](1){2} \EA[L=x_{1,3}](2){3} 
	\Edges(1,2,3)
	\end{tikzpicture}
	\caption*{Model 3: $n_2=1$, $n_1=3$.}
	\end{minipage}
	\vspace{0.1cm}
%
	\begin{minipage}[t]{.4\textwidth}
	\center
	\begin{tikzpicture}[scale=1.5] 
	\SetGraphUnit{1} 
	\SetVertexMath
	\Vertex[L=\; x_{1} \;]{1} \EA[L=x_{1,2}](1){2} \SOEA[L=x_{2,2}](1){3} \SO[L=x_{3,2}](1){4} \SOWE[L=x_{4,2}](1){5} 
	\Edges(1,2) \Edges(1,3) \Edges(1,4) \Edges(1,5)
	\end{tikzpicture}
	\caption*{Model 4: $n_2=4$, $n_1=2$.}
	\end{minipage}
	\vspace{0.1cm}
%
\begin{minipage}[t]{.5\textwidth}
	\center
	\begin{tikzpicture}[scale=1.5] 
	\SetGraphUnit{1} 
	\SetVertexMath
	\Vertex[L=\; x_{1} \;]{1} \EA[L=x_{1,2}](1){2} \EA[L=x_{1,3}](2){3} \EA[L=x_{1,4}](3){4} \EA[L=x_{1,5}](4){5} 
	\Edges(1,2,3,4,5)
	\end{tikzpicture}
	\caption*{Model 5: $n_2=1$, $n_1=5$.}
\end{minipage}
\vspace{0.1cm}
%
	\begin{minipage}[t]{.4\textwidth}
	\center
	\begin{tikzpicture}[scale=1.5] 
	\SetGraphUnit{1} 
	\SetVertexMath
	\Vertex[L=\; x_{1} \;]{1} \EA[L=x_{1,2}](1){2} \EA[L=x_{1,3}](2){3} \SOEA[L=x_{2,2}](1){4} \EA[L=x_{2,3}](4){5} 
	\Edges(1,2,3) \Edges(1,4,5)
	\end{tikzpicture}
	\caption*{Model 6: $n_2=2$, $n_1=3$.}
	\end{minipage}
	\vspace{0.1cm}
	\caption{Graphical models of the six different Hybrid Rosenbrock structures tested in this section.}
	\label{GM2drosen}
\end{figure}

Model 1 corresponds to the 2-$d$ Rosenbrock density, i.e. Equation \eqref{ndrosen-hybrid2} with $n_2=1$ and $n_1=2$, and represents the baseline against which every other model is compared. 

Model 2 ($n_2=2$, $n_1=2$) and Model 3 ($n_2=1$, $n_1=3$) are both 3-$d$ distribution. Model 2 captures the effect of extending the 2-$d$ density by adding an extra block, while Model 3 captures the effect of increasing the number of variables in the same block. We expect Model 3 to be more challenging than Model 2, as the difference between the variance of the variables in Model 3 should be higher than in Model 2. 

Model 4 ($n_2=4$, $n_1=2$) and Model 5 ($n_2=1$, $n_1=5$) are simply larger versions of Model 2 and 3, and they capture the effects that an increase in dimension of the state space has on the sampling algorithm.

Models 2 to 5 are extensions of the 2-$d$ case obtained by only increasing either the number of blocks, or the number of variables in the single block available. Model 6 ($n_2=2$, $n_1=3$) is a fully Hybrid Rosenbrock distribution, with multiple blocks and multiple variables in each block. Its 2-$d$ marginals can be seen in Figure \ref{5drosen-hyb2}.

Model 5 was included for comparison, but it is a viable option only for certain parameter values and only in low dimension. The reason is that as $n$ increases, the variance of $x_{1,n}$ grows too quickly. Using the standard parametrisation ($\mu=1$, $a=1/20$, $b=100/20$), already with $n=10$ some of the values of the sample covariance are so large that computers treat them as infinite even though they are not. Algorithms that rely on the sample covariance matrix or the Hessian to adaptively explore the target would not work properly in that case. This behaviour onsets for even lower values of $n$ if $\mu$ has a value far from zero, and $a$ is small (with respect to the original parametrisation). Hence we recommend using the block structure, to be able to increase $n$ at pleasure while avoiding uncontrolled behaviour, which is particularly undesirable in a test problem.

As the choices for the shape parameters are infinite, based on the considerations made in Section \ref{normalis} and Figure \ref{2drosen-parComparison}, we decided to test the different structures for just a few critical values of $\mu, a$ and $b$. First we will test the six models on the standard parametrisation. Then we will vary the value of the parameters one at a time, and assess the effects on the integrated autocorrelation time of one million samples from a sMMALA algorithm tuned as described above. 
The results of our experiments can be seen in Figure \ref{modelComparison1}.

\begin{figure}[!h] 
	\includegraphics[scale=0.8]{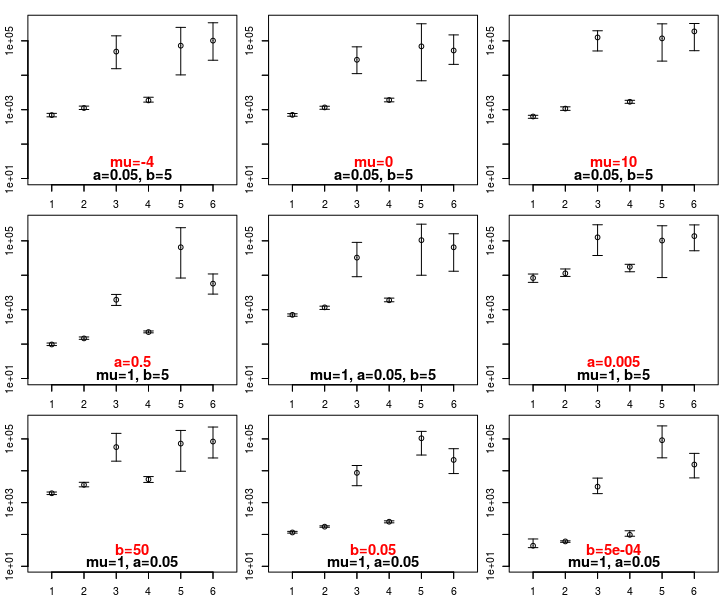}
	\centering
	\caption{Integrated autocorrelation times $\tau$ obtained varying the parameters $\mu,a,b$ for Models 1 to 6, as described in Figure \ref{GM2drosen}. The horizontal axis represents the model number, while the vertical axis show the logarithm of $\tau$. The dot represents the average value, while the whiskers represent the two-sided 95\% credibility region.}
	\label{modelComparison1}
\end{figure}

A quick inspection of the plots reveals that there is a strong tendency for Models 1, 2 and 4 to have similarly low $\tau$, with low variability, while Models 3, 5 and 6 tend to cluster on higher values of $\tau$, with higher variability. This behaviour is explained by the fact that Models 1, 2 and 4 are all parametrised by the same value of $n_1=2$. As pointed out in Section \ref{normalis}, significant differences in the scales of the components start appearing only when $n_1 \ge 3$. On the other hand, Models 3, 5 and 6 all have $n_1 \ge 3$. Therefore the difference between the variance of the various components of the state space is responsible for the increased difficulty of sampling from Models 3, 4 and 5. However, there doesn't appear to be a significant difference between Model 3, 6, both with a value of $n_1=3$, and Model 5, with a value of $n_1=5$. This suggests that increasing the value of $n_1$ beyond three does not significantly impact the overall difficulty of sampling from the model. 

\begin{figure}[!h] 
	\includegraphics[scale=0.8]{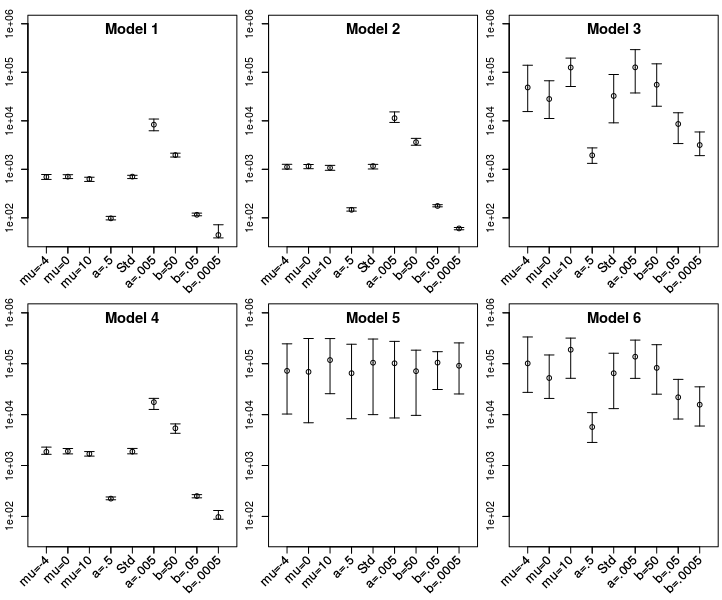}
	\centering
	\caption{Integrated autocorrelation times $\tau$ obtained varying the parameters $\mu,a,b$ for models 1 to 6, as described in Figure \ref{GM2drosen}. The values are identical to Figure \ref{modelComparison1}, but are instead grouped by model, rather than by parameter values. The horizontal axis shows the single parameter value of either $\mu,a$ or $b$ that takes a different value from the standard parametrisation, i.e. $\mu=1, a=1/20$ and $b=100/20$. The dot represents the average value, while the whiskers represent the two-sided 95\% credibility region.}
	\label{modelComparison2}
\end{figure}

Comparing Figure \ref{modelComparison1} with Figure \ref{2drosen-parComparison}, another trend becomes noticeable: models with values of $\mu,a$ and $b$ that yield rounder shapes -- i.e. large values of $a$, small values of $b$ and to a lesser extent, values of $\mu$ near zero -- tend to have lower values of $\tau$. Conversely, values of $\mu,a$ and $b$ that yield narrower and more elongated shapes -- i.e. large $a$, small $b$ and $\mu$ far from zero -- tend to have higher values of $\tau$. This is perhaps more noticeable in Figure \ref{modelComparison2}, where the data is the same as Figure \ref{modelComparison1}, but grouped by model. We have defined the parameters $\mu=1, a=1/20$ and $b=100/20$ as the ``standard parametrisation", and we have showed on the horizontal axis of Figure \ref{modelComparison2} only the single parameter of that model that differs from the standard values. That is to say, if the horizontal axis shows $\mu=-4$, the implied values of $a$ and $b$ are $a=1/20$ and $b=100/20$, while $\mu=-4$ instead of being $\mu=1$.

The one noticeable exception in Figure \ref{modelComparison2} is Model 5, which does not seem to be overly affected by varying the parameters $\mu, a$ and $b$. A likely explanation is that the variance of $x_{1,5}$ is so large that the relative effect that changing $\mu, a$ and $b$ has on the distribution is not noticeable. However, the uncertainty is quite large, so more computationally intensive tests may be needed to pinpoint the exact effects that the parameters $\mu, a$ and $b$ have on the difficulty of sampling from Model 5.

Models 1, 2 and 4 show a clear trend, with very small uncertainty: values of $a=.5, b=.05, b=.0005$ lead to lower values of $\tau$, while values of $a=.005$ and $b=50$ lead to higher values of $\tau$. Models 3 and 6 show the same trend, particularly evident on $a=.5$, albeit with larger uncertainty.

\section{Conclusions}

The 2-$d$ Rosenbrock distribution is a common benchmark problem in Markov Chain Monte Carlo sampling, when testing algorithms on densities with curved 2-$d$ marginal densities. However, as its normalising constant is not generally known, it can be hard to precisely assess the quality of the results. This is particularly true for distributions of higher dimensions, as the optimisation literature provides multiple ways to extend the 2-$d$ Rosenbrock function, but neither shape nor statistical properties of the resulting distributions are well documented. This may lead to confusion in the interpretation of the results, as it can seem that an algorithm is working appropriately while it is struggling to simulate entire regions of the distribution's support.

In this paper we have provided the normalising constant for the 2-$d$ Rosenbrock density, by splitting the density into conditional normal kernels. This property can also be used to obtain a direct Monte Carlo sample from the density. Furthermore, we showed that by carefully extending the 2-$d$ distribution to $n$ dimensions, it is possible to obtain a test problem with some very appealing features. Firstly, it is well defined and its statistical properties can be derived by simple integration. Secondly, the problem has a very challenging structure with all the 2-$d$ marginal distributions appearing as straight or curved ridges. Lastly, the variables have highly different scales, a feature that increases the difficulty of the test problem. Furthermore, we characterised the effect of the parameters on the shape of the distribution, which can be changed by algorithm developers to provide a test problem with the exact properties and dimension needed.
These properties also qualify the Hybrid Rosenbrock distribution as a good benchmark model for the computation of normalising constants.

Finally, we have verified the accuracy of the proposed distributions and their adequacy as a challenging benchmark problem with numerical experiments, where the performance assessment was made tremendously more accurate by the availability of analytic solutions for the Hybrid Rosenbrock density.


\section*{Acknowledgements}
FP would like to thank Dr. Tim Waite and Dr. Simon Cotter for their many useful comments, and the Department of Mathematics at the University of Manchester for PhD funding.

\clearpage

\appendix

\section{Details of proof of Proposition \ref{constHybridProposition} }
\label{proof:constHybridProposition}

The integral of Equation \eqref{ndrosen-hybrid2} over the domain $\mathbb{R}^n$ is

\begin{align}
\int_{\mathbb{R}^n} \exp & \left \{ - a (x_{1}-\mu)^2 - \sum_{j=1}^{n_2} \sum_{i=2}^{n_1} b_{j,i} (x_{j,i} - x_{j,i-1}^2)^2 \right \}  \, dx_{n_2,n_1} \dots dx_1= \\
& = \int_\mathbb{R}  \exp \left \{ -a (x_1-\mu)^2 \right \} \prod_{j = 1}^{n_2} \prod_{i = 2}^{n_1} \int_\mathbb{R} \exp \left \{ - b_{j,i} (x_{1,i}-x_{1,i-1}^2)^2 \right \} \, dx_{n_2,n_1} \cdots \, dx_{1} \, , \nonumber
\end{align}
by splitting the terms in the exponential function. We can now isolate the last integral, with indices $j=n_2$ and $i=n_1$, as
\begin{align}
\label{isolateIntegral1mv}
= \int_\mathbb{R} \exp & \left \{ -a (x_1-\mu)^2 \right \} \prod_{j = 1}^{n_2} \prod_{i = 2}^{n_1-1} \int_\mathbb{R} \exp \left \{ - b_{j,i} (x_{1,i}-x_{1,i-1}^2)^2 \right \} \times \\
& \times \int_\mathbb{R} \exp \left \{ - b_{n_2,n_1} (x_{n_2,n_1}-x_{n_2,n_{n_1}-1}^2)^2 \right \} \, dx_{n_2,n_1} dx_{n_2,n_{1}-1} \cdots \, dx_{1} \, . \nonumber
\end{align}
From Proposition \ref{2drosen-proposition} we know that by changing variables $v_{n_2,n_1} = x_{n_2,n_1}-x_{n_2,n_1}^2$,
\begin{equation*}
\int_\mathbb{R} \exp \left \{ - b_{n_2,n_1} (x_{n_2,n_1}-x_{n_2,n_{n_1}-1}^2)^2 \right \} \, dx_{n_2,n_1} = \sqrt{\frac{\pi}{b_{n_2,n_1}}} \, .
\end{equation*}
We can substitute this result back into \eqref{isolateIntegral1mv}, which becomes
\begin{equation}
\label{isolateIntegral2mv}
= \sqrt{\frac{\pi}{b_{n_2,n_1}}} \int_\mathbb{R} \exp  \left \{ -a (x_1-\mu)^2 \right \} \prod_{j = 1}^{n_2} \prod_{i = 2}^{n_1-1} \int_\mathbb{R} \exp \left \{ - b_{j,i} (x_{1,i}-x_{1,i-1}^2)^2 \right \} dx_{n_2,n_{1}-1} \cdots \, dx_{1} \, .
\end{equation}
We can apply the same procedure to all the integrals in Equation \eqref{isolateIntegral2mv} in turn, starting from the remaining last variable $n_1-1$ of the last block $n_2$, until the very first variable $x_1$. This operation yields
\begin{equation*}
\int_{\mathbb{R}^n} \exp \left \{ - a (x_{1}-\mu)^2 - \sum_{j=1}^{n_2} \sum_{i=2}^{n_1} b_{j,i} (x_{j,i} - x_{j,i-1}^2)^2 \right \}  \, dx_{n_2,n_1} \dots dx_1 = \frac{\pi^{n/2}}{\sqrt{a} \; \prod_{i=2,j=1}^{n_1,n_2} \sqrt{b_{j,i}} } \, .
\end{equation*}
Taking the reciprocal we obtain the normalisation constant in the statement of Proposition \ref{constHybridProposition}.

\end{document}